  \documentclass[pagebackref]{article}
  \usepackage[a4paper]{geometry}

\usepackage[utf8]{inputenc}
\usepackage[T1]{fontenc}
\usepackage{xcolor}
\usepackage{amssymb,amsmath}
\usepackage{amsthm}
  \usepackage{graphicx}
\usepackage{booktabs}
\usepackage{tabularx}
\usepackage{multicol}
  \usepackage[numbers,sort&compress]{natbib} 
\usepackage{tikz}
\usetikzlibrary{calc,shapes,positioning}

\usepackage{mathtools}

\usepackage{contour}

\usepackage{xargs} 
\newcommandx{\set}[2][1=1]{\ensuremath{\{#1,\ldots,#2\}}}
\newcommandx{\tlog}[3][1=,3=]{\log_{#1}^{#3}(#2)}

  \usepackage[linkcolor=red!50!black,citecolor=green!40!black,colorlinks]{hyperref}
  \usepackage[capitalize,nameinlink]{cleveref}

  \newtheorem{theorem}{Theorem}
  \newtheorem{lemma}{Lemma}
  \newtheorem{proposition}{Proposition}
  
  \newtheorem{observation}{Observation}
  \newtheorem{fact}{Fact}
  \theoremstyle{definition}
  \newtheorem{problem}{Problem}
  \newtheorem{question}{Question}
  \newtheorem{definition}{Definition}
  
  \newtheorem{construction}{Construction}
\crefname{observation}{Observation}{Observations}
\crefname{rrule}{Reduction Rule}{Reduction Rules}
\crefname{construction}{Construction}{Constructions}
\Crefname{proposition}{Prop.}{Props.}
\crefname{proposition}{Proposition}{Propositions}
\Crefname{theorem}{Thm.}{Thm.}
\crefname{theorem}{Theorem}{Theorems}
\Crefname{lemma}{Lem.}{Lems.}
\crefname{lemma}{Lemma}{Lemmas}
\Crefname{corollary}{Cor.}{Cors.}
\crefname{corollary}{Corollary}{Corollaries}
\Crefname{fact}{Fact}{Facts}
\crefname{fact}{Fact}{Facts}
  \crefname{figure}{Figure}{Figures}

\newcommand{\calQ}{\mathcal{Q}}

\crefname{problem}{Problem}{Problems}
\Crefname{problem}{Prob.}{Probs.}
\newcommandx{\decprob}[6][3=Input,5=Question]{
  \begingroup
  \par\noindent\nopagebreak[4]
  \begin{problem}\label{prob:#2}\vspace{-0.5em}\colorbox{gray!17!white}{\textsc{#1}\index{problem!#1}}\nopagebreak[4]\end{problem}\nopagebreak[4]\vspace{-0.5em}
  \par\noindent\hangindent=\parindent\textbf{#3}:  #4\nopagebreak[4]
  \par\noindent\hangindent=\parindent\textbf{#5}:  #6
  \par\bigskip
  \endgroup
}

\newcommand{\N}{\mathbb{N}}
\newcommand{\Nzero}{\mathbb{N}_0}

\newcommand{\prob}[1]{\textnormal{\textsc{#1}}}
\newcommand{\fvsTsc}{\prob{Feedback Vertex Set}}

\newcommand{\tcTsc}{\prob{3-Coloring}}
\newcommand{\tcAcr}{\prob{3-Col}}
\newcommand{\isTsc}{\prob{Independent Set}}

\newcommand{\dsTsc}{\prob{Dominating Set}}

\newcommand{\cocl}[1]{\ensuremath{\operatorname{#1}}}

\newcommand{\NP}{\cocl{NP}}

\newcommand{\cqed}{\hfill$\diamond$}
\newcommand{\tref}[1]{(\Cref{#1})}

\newcommand{\ceq}{\ensuremath{\coloneqq}}

\newcommand{\lqed}{}

\newcommand{\ceil}[1]{\lceil#1\rceil}
\newcommand{\CS}{\{1,2,3\}}
\newcommand{\HC}{\mathcal{C}}

\usepackage{etoolbox}

\newcommand{\ExternalLink}{%
  \tikz[x=1.2ex, y=1.2ex, baseline=-0.05ex]{%
    \begin{scope}[x=1ex, y=1ex]
        \clip (-0.1,-0.1) 
            --++ (-0, 1.2) 
            --++ (0.6, 0) 
            --++ (0, -0.6) 
            --++ (0.6, 0) 
            --++ (0, -1);
        \path[draw, 
            line width = 0.5, 
            rounded corners=0.5] 
            (0,0) rectangle (1,1);
    \end{scope}
    \path[draw, line width = 0.5] (0.5, 0.5) 
        -- (1, 1);
    \path[draw, line width = 0.5] (0.6, 1) 
        -- (1, 1) -- (1, 0.6);
  }
}

\newcommand{\tikzpramble}{%
  \def\teps{0.33}
  \def\nsc{0.6}
  \tikzstyle{xnode}=[circle,scale=\nsc,draw,fill=white];
  \tikzstyle{xnodeA}=[circle,scale=\nsc,fill=white,draw=red];
  \tikzstyle{xnodeB}=[circle,scale=\nsc,fill=lightgray,draw=green];
  \tikzstyle{xnodeC}=[circle,scale=\nsc,fill=black,draw=blue];
  \tikzstyle{xnodex}=[circle,fill,scale=\nsc,draw];
  \tikzstyle{xnodey}=[diamond,fill,scale=\nsc,draw];
  \tikzstyle{xedge}=[thick,-];
  \tikzstyle{xedgex}=[thick,-,dashed];
  \tikzstyle{xedgedot}=[thick,-,dotted];
  
  \tikzstyle{xpath}=[color=blue,opacity=0.15,line cap=round,line width=6pt];
  \tikzstyle{xpathx}=[color=magenta,opacity=0.15,line cap=round,line width=6pt];
  \tikzstyle{xpathy}=[color=green,opacity=0.15,line cap=round,line width=6pt];
  
  \tikzstyle{xhili}=[circle,scale=1.25,opacity=0.25,fill,color=orange,draw];
  \tikzstyle{xhiliIS}=[circle,scale=1.25,opacity=0.25,fill,color=magenta,draw];
  \tikzstyle{xxhili}=[scale=1.25,opacity=0.25,fill,color=cyan,draw];
}

\newcommandx{\theLad}[5][2=-1,5=1]{
     \node (x#1) at (0+#3*\xr,0+#4*\yr)[xnode]{};
     \node (#1x1) at (#5*0.5*\xr+#3*\xr,#2*1*\yr+#4*\yr)[xnode]{};
     \node (#1x2) at (#5*1*\xr+#3*\xr,#2*0.5*\yr+#4*\yr)[xnode]{};
     \node (#1x3) at (#5*1.5*\xr+#3*\xr,#2*1*\yr+#4*\yr)[xnode]{};
     \node (#1x4) at (#5*2*\xr+#3*\xr,#2*0.5*\yr+#4*\yr)[xnode]{};
     \node (#1x5) at (#5*2.5*\xr+#3*\xr,#2*1*\yr+#4*\yr)[xnode]{};
     \node (#1x6) at (#5*3*\xr+#3*\xr,#2*0.5*\yr+#4*\yr)[xnode]{};
     \node (y#1) at (#5*3.5*\xr+#3*\xr,0*\yr+#4*\yr)[xnode]{};
     \foreach \x/\y in {x#1/#1x1,x#1/#1x2,#1x1/#1x2,#1x1/#1x3,#1x2/#1x3,#1x2/#1x4,#1x3/#1x4,#1x3/#1x5,#1x4/#1x5,#1x4/#1x6,#1x5/#1x6,#1x6/y#1}
     {\draw[xedge] (\x) to (\y);}
     
}

\newcommand{\theH}[2]{%
  \theLad{a}{0+#1*\xr}{0+#2*\yr}
  \theLad{b}[1]{0+#1*\xr}{-3*\yr+#2*\yr}
  \foreach\x\y in{ax1/bx1,ax5/bx5,ax6/bx6}{\draw[xedge] (\x) to (\y);}
}

\newcommand{\theHcolor}{%
  \foreach\x\y in{xa/xnodeA,ax1/xnodeB,ax2/xnodeC,ax3/xnodeA,ax4/xnodeB,ax5/xnodeC,ax6/xnodeA,xb/xnodeB,bx1/xnodeA,bx2/xnodeC,bx3/xnodeB,bx4/xnodeA,bx5/xnodeC,bx6/xnodeB}
     {\node at (\x)[\y]{};}
}
\newcommand{\theHpath}{%
  \draw[xpath] (yb) to (bx6) to (ax6) to (ya);
  \draw[xpath] (xb) to (bx2) to (bx4) to (bx5) to (bx3) to (bx1) to (ax1) to (ax3) to (ax5) to (ax4) to (ax2) to (xa);
}
\newcommand{\theS}[2]{%
  \theLad{a}{0}{0}
  \theLad{b}[1]{3.5}{-3}[-1]
  \node (ab1) at (0.5*\xr,-1.5*\yr)[xnode]{};
  \node (ab2) at (3*\xr,-1.5*\yr)[xnode]{};
  \foreach\x\y in{ab1/ax1,ab1/bx6,ab1/ax5,ab1/ab2,ab2/bx1,ab2/ax6,ab2/bx5}{\draw[xedge] (\x) to (\y);}
}
\newcommand{\theScolor}{%
  \foreach\x\y in{xa/xnodeA,ax1/xnodeB,ax2/xnodeC,ax3/xnodeA,ax4/xnodeB,ax5/xnodeC,ax6/xnodeA,xb/xnodeB,bx1/xnodeA,bx2/xnodeC,bx3/xnodeB,bx4/xnodeA,bx5/xnodeC,bx6/xnodeB,ab1/xnodeA,ab2/xnodeB}
     {\node at (\x)[\y]{};}
}
\newcommand{\theSpath}{%
  \draw[xpath] (xb) to (bx1) to (bx2) to (bx3) to (bx4) to (bx5) to (ab2) to (ax6) to (ya);
     \draw[xpath] (xa) to (ax1) to (ax2) to (ax3) to (ax4) to (ax5) to (ab1) to (bx6) to (yb);
}

\newcommand{\theDailey}[3]{%
     \node (#1d1) at (0+#2*\xr,0)[xnode]{};
     \node (#1c1) at (0+#2*\xr,-0.5*\yr+#3*\yr)[xnode]{};
     \node (#1c2) at (-0.5*\xr+#2*\xr,0*\yr+#3*\yr)[xnode]{};
     \node (#1c3) at (-0.25*\xr+#2*\xr,0.5*\yr+#3*\yr)[xnode]{};
     \node (#1c4) at (0.25*\xr+#2*\xr,0.5*\yr+#3*\yr)[xnode]{};
     \node (#1c5) at (0.5*\xr+#2*\xr,0*\yr+#3*\yr)[xnode]{};
     \draw[xedge] (#1c3) to (#1c2) to (#1c1) to (#1c5) to (#1c4);
     \foreach\x in{1,...,5}{\draw[xedge] (#1d1) to (#1c\x);}
     \node (#1b1) at (0*\xr+#2*\xr,-1*\yr+#3*\yr)[xnode]{};
     \node (#1b2) at (-0.75*\xr+#2*\xr,-0.75*\yr+#3*\yr)[xnode]{};
     \node (#1b3) at (-1*\xr+#2*\xr,-0*\yr+#3*\yr)[xnode]{};
     \node (#1b4) at (-0.75*\xr+#2*\xr,0.5*\yr+#3*\yr)[xnode]{};
     \node (#1b5) at (-0.5*\xr+#2*\xr,1*\yr+#3*\yr)[xnode]{};
     \node (#1b6) at (0*\xr+#2*\xr,1*\yr+#3*\yr)[xnode]{};
     \node (#1b7) at (0.5*\xr+#2*\xr,1*\yr+#3*\yr)[xnode]{};
     \node (#1b8) at (0.75*\xr+#2*\xr,0.5*\yr+#3*\yr)[xnode]{};
     \node (#1b9) at (1*\xr+#2*\xr,-0*\yr+#3*\yr)[xnode]{};
     \node (#1b10) at (0.75*\xr+#2*\xr,-0.75*\yr+#3*\yr)[xnode]{};
     \foreach\x in{1,...,9}{\pgfmathsetmacro\yx{int(\x + 1)};\draw[xedge] (#1b\x) to (#1b\yx);}
     \draw[xedge] (#1b10) to (#1b1);
     \draw[xedge] (#1b1) to (#1b3);
     \draw[xedge] (#1b1) to (#1c1) to (#1b10);
     \draw[xedge] (#1b3) to (#1c2) to (#1b4);
     \draw[xedge] (#1b4) to (#1c3) to (#1b5);
     \draw[xedge] (#1c3) to (#1b6) to (#1c4);
     \draw[xedge] (#1b7) to (#1c4) to (#1b8);
     \draw[xedge] (#1b8) to (#1c5) to (#1b10);
     \node (#1a1) at (0*\xr+#2*\xr,-1.75*\yr+#3*\yr)[xnode]{};
     \node (#1a2) at (-1.75*\xr+#2*\xr,-1.25*\yr+#3*\yr)[xnode]{};
     \node (#1a3) at (-1.5*\xr+#2*\xr,0*\yr+#3*\yr)[xnode]{};
     \node (#1a4) at (-1.75*\xr+#2*\xr,1.25*\yr+#3*\yr)[xnode]{};
     \node (#1a5) at (0*\xr+#2*\xr,1.5*\yr+#3*\yr)[xnode]{};
     \node (#1a6) at (0*\xr+#2*\xr,2*\yr+#3*\yr)[xnode]{};
     \node (#1a7) at (1.*\xr+#2*\xr,1.*\yr+#3*\yr)[xnode]{};
     \node (#1a8) at (1.75*\xr+#2*\xr,0*\yr+#3*\yr)[xnode]{};
     \node (#1a9) at (1.*\xr+#2*\xr,-1.*\yr+#3*\yr)[xnode]{};
     \foreach\x in{1,...,8}{\pgfmathsetmacro\yx{int(\x + 1)};\draw[xedge] (#1a\x) to (#1a\yx);}
     \draw[xedge] (#1a1) to (#1a9);
     \foreach\x in{#1a1,#1a2,#1a3}{\draw[xedge] (#1b2) to (\x);}
     \foreach\x in{#1b3,#1b4}{\draw[xedge] (#1a3) to (\x);}
     \foreach\x in{#1a4,#1a5}{\draw[xedge] (#1b5) to (\x);}
     \foreach\x in{#1b6,#1b7}{\draw[xedge] (#1a5) to (\x);}
     \foreach\x in{#1b7,#1b8,#1b9}{\draw[xedge] (#1a7) to (\x);}
     \draw[xedge] (#1a8) to (#1b9);
     \foreach\x in{#1b9,#1b10,#1b1}{\draw[xedge] (#1a9) to (\x);}
     \draw[xedge] (#1a2) to (#1a4);
     \draw[xedge] (#1a4) to (#1a6);
     \draw[xedge] (#1a6) to [out=0,in=90](#1a8);
     \draw[xedge] (#1a1) to [out=0,in=-90](#1a8);
}

\newcommand{\theDD}[2]{%
  \begin{scope}[xscale=-1]
    \theDailey{A}{0+#1*\xr}{0+#2*\yr}
    \node[below =of Aa1] (Ax) [xnode]{};
    \draw[xedge] (Aa1) to (Ax);
  \end{scope}
  \theDailey{B}{5*\xr+#1*\xr}{0+#2*\yr}
  \node[below =of Ba1] (Bx) [xnode]{};
  \draw[xedge] (Ba1) to (Bx);
  \draw[xedge] (Aa2) to (Ba2);
  \draw[xedge] (Aa6) to (Ba6);
}

\newcommand{\theDDx}[2]{%
  \begin{scope}[xscale=-1]
    \theDailey{A}{0+#1*\xr}{0+#2*\yr}
  \end{scope}
  \theDailey{B}{5*\xr+#1*\xr}{0+#2*\yr}
  \draw[xedge] (Aa2) to (Ba2);
  \draw[xedge] (Aa6) to (Ba6);
}

\newcommand{\theDDpathS}[2]{
  \draw[xpath] (#1a1) to (#1b2) to (#1a2) to (#1a3) to (#1a4) to (#1a5) to (#1b5) to (#1b4) to (#1b3) to (#1b1) to (#1a9) to (#1b10)
  to (#1c1) to (#1c2) to (#1c3) to (#1b6) to (#1b7) to (#1c4) to (#1d1) to (#1c5) to (#1b8) to (#1a7) to (#1b9) to (#1a8) to [out=90,in=#2](#1a6);
}

\newcommand{\theDDpath}{
  \theDDpathS{A}{180}
  \theDDpathS{B}{0}
  \draw[xpath] (Aa6) to (Ba6);
}

\newcommand{\theEllipse}[3]{
  \draw[fill=white,opacity=0.7] (0,0) ellipse (#1*1 and #2*1);
  \draw (0,0) ellipse (#1*1 and #2*1);
  \foreach\x in {0,...,5}{\node (#3a\x) at ($(0,0)+(180/5*\x:#1*1 and #2*1)$)[xnode]{};}
  \foreach\x in {1,...,2}{\node (#3b\x) at ($(0,0)+(-60*\x:#1*1 and #2*1)$)[xnode]{};}
  \node (#3cd) at ($(0,0)+(270:#1*1 and #2*1)$)[fill=white]{$\cdots$};
  \draw[draw=none] (0,0) ellipse (#1*1*0.75 and #2*1*0.75);
  \foreach\x in {0,...,19}{\node (#3c\x) at ($(0,0)+(18*\x:#1*1*0.75 and #2*1*0.75)$)[inner sep=0pt]{};}
  \draw[xedgex] (#3a0) to (#3c0);
  \draw[xedgex] (#3a0) to (#3c1);
  \draw[xedgex] (#3a1) to (#3c2);
  \draw[xedgex] (#3a1) to (#3c3);
  \draw[xedgex] (#3a2) to (#3c4);
  \draw[xedgex] (#3a2) to (#3c5);
  \draw[xedgex] (#3a3) to (#3c6);
  \draw[xedgex] (#3a3) to (#3c7);
  \draw[xedgex] (#3a4) to (#3c8);
  \draw[xedgex] (#3a4) to (#3c9);
  \draw[xedgex] (#3a5) to (#3c10);
  \draw[xedgex] (#3a5) to (#3c11);
  \draw[xedgex] (#3b2) to (#3c12);
  \draw[xedgex] (#3b2) to (#3c13);
  \draw[xedgex] (#3b1) to (#3c17);
  \draw[xedgex] (#3b1) to (#3c18);
  
}

\newcommand{\mytitle}{3-Coloring on Regular, Planar, and Ordered Hamiltonian~Graphs}

\newcommand{\myabstract}{%
We prove that \tcTsc{}
remains \NP-hard on 4- and 5-regular planar Hamiltonian graphs,
strengthening the results of 
Dailey~[Disc.~Math.'80]
and 
Fleischner and Sabidussi~[J.\ Graph.\ Theor.'02].
Moreover,
we prove that \tcTsc{} remains \NP-hard on~$p$-regular Hamiltonian graphs for every~$p\geq 6$ and~$p$-ordered regular Hamiltonian graphs for every~$p\geq 3$.
}

  \makeatletter
  \def\abstractname{Abstract.}
  \renewenvironment{abstract}{%
      \if@twocolumn
        \section*{\abstractname}%
      \else
        \small
        \quotation
	\noindent{\bfseries\abstractname}%
      \fi}
      {\if@twocolumn\else\endquotation\fi}
  \makeatother
  \title{\Large \bf \mytitle}
  \author{Dario~Cavallaro \and Till Fluschnik\footnote{Supported by DFG, project TORE (NI/369-18).}}
  \date{\small Technische Universität Berlin, Faculty~IV,\\ Algorithmics and Computational Complexity, Germany.\\\texttt{cavallaro@campus.tu-berlin.de,till.fluschnik@tu-berlin.de}}

\begin{document}
  \maketitle
\begin{abstract}
\myabstract{}

\end{abstract}
\section{Introduction}
We study the computational complexity 
of the following 
classic \NP-complete
problem~\cite{Karp72} on Hamiltonian graphs 
(graphs
admitting a cycle that visits every vertex
exactly once)
that are additionally
restricted to be planar, 
regular,
and ordered.

\decprob{\tcTsc{} (\tcAcr{})}{tc}
{An undirected graph~$G=(V,E)$.}
{Is $G$ 3-colorable?}

\noindent
\cref{fig:results} gives an overview of our results.
\begin{figure}[t!]
 \centering
  \begin{tikzpicture}
 
    \usetikzlibrary{patterns,backgrounds,shapes}
    \usetikzlibrary{calc}

      \def\yr{1}
      \def\xr{1}
      \def\boxw{2.9}
    \def\boxh{1.5}

    \def\colNP{orange!50!red!18!white}
    \def\colP{green!20!white}
    \def\colOpen{blue!20!white}
      \def\fsres{\footnotesize}
      \def\fsresx{\scriptsize}

    \def\bwA{1}

    \tikzstyle{xarc}=[->,gray,thick,>=latex,rounded corners]
    
    \newcommand{\gbox}[6]{
      \node (a#1) at (#2)[rectangle, rounded corners, minimum width=\xr*\boxw cm, minimum height=\yr*\boxh cm,text width=\xr*\boxw cm,fill=white,very thick,draw,align=center]{#3#4};
    }

    \def\ysh{1.9}
    \def\xshA{1.625}
    \def\xshB{4.9}
    \def\teps{0.1}

    \newcommand{\theDiagram}{
      \gbox{h}{0.0*\xr,-0.5*\yr}{Hamiltonian}{}{\colNP};
      \gbox{conn}{-\xshB*\xr,-0.5*\ysh*\yr}{$(p-1)$-connected Hamiltonian}{}{\colNP};
      \gbox{3reg}{\xshB*\xr,-\ysh*\yr}{${\leq3}$-regular Hamiltonian}{ \cite{Brooks41}}{\colP};
      \gbox{reg}{-\xshA*\xr,-1.25*\ysh*\yr}{$p$-regular Hamiltonian, $p\geq 6$}{ \tref{thm:tc:kreg}}{\colNP};
      \gbox{ord}{-\xshB*\xr,-1.5*\ysh*\yr}{$p$-ordered  Hamiltonian$^*$}{ \tref{thm:tc:pordham}}{\colNP};
      \gbox{plan}{\xshA*\xr,-1.25*\ysh*\yr}{planar Hamiltonian}{}{\colNP};
      \gbox{ordham}{-\xshB*\xr,-2.5*\ysh*\yr}{$p$-Hamiltonian-ordered}{}{\colNP};
      \gbox{plan3reg}{\xshB*\xr,-2.5*\ysh*\yr}{planar ${\leq3}$-regular Hamiltonian}{}{\colP};
      \gbox{plan5reg}{-\xshA*\xr,-2.5*\ysh*\yr}{planar 5-regular Hamiltonian}{ \tref{thm:tc:5regplanham}}{\colNP};
      \gbox{plan4reg}{\xshA*\xr,-2.5*\ysh*\yr}{planar 4-regular Hamiltonian}{ \tref{thm:tc:4regplanaHam}}{\colNP};
      \draw[xarc] (aord) to node[midway,right]{\footnotesize\cite{NgS97}}(aconn);
      \draw[xarc] (aordham) to node[midway,right]{
      }(aord);
      \draw[xarc] (aplan3reg) to (aplan);
      \draw[xarc] (aplan3reg) to (a3reg);
      \draw[xarc] (aplan4reg) to (aplan);
      \draw[xarc] (aplan5reg) to (aplan);
      
      \draw[xarc] (a3reg) to (ah);
      \draw[xarc] (areg) to (ah);
      \draw[xarc] (aconn) to (ah);
      \draw[xarc] (aplan) to (ah);
    }
    \theDiagram{}
    \draw[rounded corners,dashed,fill=\colNP] ($(ah.north west)+(-\teps,+\teps)$) 
    -- ($(ah.west|-aconn.north)+(-\teps,+\teps)$)
    -- ($(aconn.north west)+(-\teps,+\teps)$)
    -- ($(aord.south west)+(-\teps,-\teps)$)
    -- ($(aplan5reg.west|-aord.south)+(-\teps,-\teps)$)
    -- ($(aplan5reg.south west)+(-\teps,-6*\teps)$)
    -- ($(aplan4reg.south east)+(+\teps,-6*\teps)$)
    -- ($(aplan.north east)+(+\teps,+\teps)$)
    -- ($(ah.east|-aplan.north)+(+\teps,+\teps)$)
    -- ($(ah.north east)+(+\teps,+\teps)$)
    -- cycle;
    \draw[rounded corners,dashed,fill=\colP] ($(a3reg.north west)+(-\teps,+\teps)$) 
    -- ($(aplan3reg.south west)+(-\teps,-6*\teps)$)
    -- ($(aplan3reg.south east)+(+\teps,-6*\teps)$)
    -- ($(a3reg.north east)+(+\teps,\teps)$)
    -- cycle;
    \draw[rounded corners,dashed,fill=gray!15!white] ($(aordham.north west)+(-\teps,+\teps)$) 
    -- ($(aordham.south west)+(-\teps,-6*\teps)$) 
    -- ($(aordham.south east)+(+\teps,-6*\teps)$)
    -- ($(aordham.north east)+(+\teps,+\teps)$) 
    -- cycle;
    
    \theDiagram{}
    \node at (aplan5reg.south west)[anchor=north west,yshift=-0.25*\yr em]
    {\textsl{\NP-hard}};
    \node at (aplan3reg.south west)[anchor=north west,yshift=-0.25*\yr em]{\textsl{Polynomial~time}};
    \node at (aordham.south west)[anchor=north west,yshift=-0.25*\yr em]{\textsl{Open}};

  \end{tikzpicture}
  \caption{Overview of our results.
  An arrow from a box~$A$ to a box~$B$ describes that~$A$'s graph class is included in $B$'s graph class.
  All shown \NP-hardness results hold true
  even if a Hamiltonian cycle is provided as part of the input.
  $^*$\,(holds still true if the graph is additionally regular)
  }
  \label{fig:results}
\end{figure}
\citet{Dailey80} 
proved that \tcTsc{} is \NP-complete 
on 4- and 5-regular planar graphs,
and 
\citet{FleischnerS03}
proved that \tcTsc{} is \NP-complete 
and on 4-regular Hamiltonian graphs.
We strengthen these two results 
by proving that 
\tcTsc{} is \NP-hard 
on 4- and 5-regular planar Hamiltonian graphs
(\cref{sec:planreg}).
Note that \tcTsc{} is polynomial-time solvable on graphs of maximum degree three~\cite{Brooks41}.
Moreover,
we prove that \tcTsc{} 
is \NP-hard
on~$p$-regular Hamiltonian graphs for every~$p\geq 6$
(\cref{sec:reg}).
Finally,
we prove that \tcTsc{} 
is \NP-hard 
on $p$-ordered regular Hamiltonian graphs for every~$p\geq 3$,
which implies \NP-hardness on regular $(p-1)$-connected Hamiltonian graphs
(\cref{sec:ordcon}).
All our \NP-hardness results still hold true
if a Hamiltonian cycle is additionally provided as part of the input
(recall that computing a Hamiltonian cycle is \NP-complete in general~\cite{Karp72}).

  \subparagraph{Related Work.}

In our previous work~\cite{CavallaroF21},
we proved that
\fvsTsc{}
remains \NP-complete on
4- and 5-regular planar Hamiltonian graphs,
$p$-regular Hamiltonian graphs for every~$p\in\N_{\geq 6}$,
and~$p$-Hamiltonian-ordered graphs for every~$p\in\N_{\geq 3}$.
Several classic graph problems are studied on planar regular graphs.
For instance,
\isTsc{} and \dsTsc{} are proven to be \NP-complete on planar 3-regular graphs~\cite{Mohar01,kikuno1980np}.
Moreover,
\isTsc{}
is proven to be \NP-hard 
on 3- and 4-regular Hamiltonian graphs~\cite{FleischnerSS10}.
$p$-Hamiltonian-ordered graphs were introduced by~\citet{NgS97}.

\section{Preliminaries}
\label{sec:prelims}

We denote by~$\N$ and~$\Nzero$ the natural numbers excluding and including zero,
respectively.
We denote by~$\N_{\geq p}$ with~$p\in\N$ the set~$\N_0\setminus\set[0]{p-1}$.
  \subparagraph*{Graph Theory.}

We use basic notations from graph theory~\cite{Diestel10}.
We call graph~$G=(V,E)$ even if~$|V|$ is even,
and odd otherwise.
The neighborhood~$N_G(v)$ of a vertex~$v\in V$ in~$G$
is the vertex set~$\{w\in V\mid \{v,w\}\in E\}$.
The degree of a vertex~$v$ is~$|N_G(v)|$.
A graph is ($k$-)regular if every vertex has the same degree ($k$).
A cycle is a 2-regular connected graph.
We represent a cycle also by a tuple,
that is,
the cycle~$C=(v_0,\dots,v_{n-1},v_0)$ is the cycle with vertex set~$\{v_0,\dots,v_{n-1}\}$
and edge set~$\{\{v_i,v_{i+1\bmod n}\} \mid i\in\set[0]{n-1}\}$.
Let~$v,w$ be two distinct vertices in~$G=(V,E)$.
The graph obtained by \emph{identifying~$v$ with~$w$}
has vertex set~$(V\setminus\{v,w\})\cup\{vw\}$,
where~$vw$ is a new vertex,
and
edge set~$(E\setminus\{e\in E\mid \{v,w\}\cap e\neq\emptyset\})\cup\{\{vw,x\}\mid x\in (N_G(v)\cup N_G(w))\setminus\{v,w\}\}$.
A graph is
planar 
if it can be drawn on the two-dimensional plane 
with no two edges crossing except at their endpoints.
A graph is~$p$-(Hamiltonian-)ordered,
$p\in\N_{\geq 3}$
if for each $p$-tuple~$(v_1,\dots,v_p)$ of vertices
there exists a (Hamiltonian) cycle that visits the vertices~$x_1,\dots,x_p$ in this order.

A $k$-coloring of~$G=(V,E)$ is a function~$f\colon V\to \set{k}$.
A $k$-coloring~$f$ of~$G$ is valid if for every~$\{v,w\}\in E$ it holds that~$f(v)\neq f(w)$.
Graph~$G$ is~$k$-colorable if there is a valid~$k$-coloring of~$G$.
For~$\emptyset\neq V'\subset V$,
a $k$-coloring~$f'\colon V'\to\set{k}$ is called partial.
A partial $k$-coloring~$f'\colon V'\to \set{k}$ is called extendable
if there is a partial~$k$-coloring~$g\colon V\setminus V'\to\set{k}$
such that the~$k$-coloring~$f\colon V\to\set{k}$ with~$f(v)=f'(v)$ if~$v\in V'$,
and~$f(v)=g(v)$ if~$v\in V\setminus V'$,
is valid.

\section{Planar Regular Hamiltonian Graphs}
\label{sec:planreg}

In this section,
we prove that \tcTsc{} 
is \NP-hard on 4- and 5-regular planar Hamiltonian graphs
with a Hamiltonian cycle provided.
Note that no planar $p$-regular graph with~$p\geq 6$ exists.

\subsection{4-regular planar Hamiltonian}
\label{ssec:4regplanarham}

We prove the following.

\begin{theorem}%
 \label{thm:tc:4regplanaHam}
 \tcTsc{} on 4-regular planar Hamiltonian graphs
 is \NP-hard,
 even if a Hamiltonian cycle is provided.
\end{theorem}

\noindent
Following the idea of Fleischner and Sabidussi~\cite{FleischnerS03},
we first compute a 2-factor in polynomial time.
A 2-factor of a graph~$G$ is a set~$\{Q_1,\dots,Q_t\}$,
$t\in\N$,
of 
cycles such that~$V(G)=\bigcup_{i\in\set{t}} V(Q_i)$.
Then,
we iteratively ``merge'' two cycles from the 2-factor 
via one of the graphs~$H$ and~$S$
(see~\cref{fig:gtb}).
\begin{figure}[t]
 \centering
  \begin{tikzpicture}

    \def\xr{1.25}
    \def\yr{0.75}
    \tikzpramble{};
    
    \begin{scope}[xshift=0.5cm,yshift=-2.5*\yr cm]
     \node at (-0.5*\xr,0.5*\yr)[]{(a)};
     \node at (1.0*\xr,0.45*\yr)[anchor=west]{Graph~$L$};
     \theLad{a}{0}{0}
     \foreach\x\y in{xa/xnodeA,ax1/xnodeB,ax2/xnodeC,ax3/xnodeA,ax4/xnodeB,ax5/xnodeC,ax6/xnodeA}
     {\node at (\x)[\y]{};}
     \node at (xa)[label=90:{$x$}]{};
     \node at (ya)[label=90:{$y$}]{};
     \node at (ax1)[label=180:{$x_1$}]{};
     \node at (ax6)[label=90:{$y_1$}]{};
     \node at (ax5)[label=0:{$y_2$}]{};
     \node at (ax2)[label=90:{$x_2$}]{};
     \node at (ax4)[label=90:{$y_3$}]{};
     \node at (ax3)[label=-90:{$z$}]{};
    \end{scope}
    
    \begin{scope}[xshift=-2.5cm,yshift=-5*\yr cm]
     \node at (-0.5*\xr,0.5*\yr)[]{(b)};
     \node at (1.0*\xr,0.25*\yr)[anchor=west]{Graph~$H$};
     \theH{0}{0}
     \theHcolor{}
     \theHpath{}
     \node at (xa)[label=90:{$x$}]{};
     \node at (ya)[label=90:{$y$}]{};
     \node at (ax6)[label=90:{$y_1$}]{};
     \node at (xb)[label=-90:{$x'$}]{};
     \node at (yb)[label=-90:{$y'$}]{};
     \node at (bx6)[label=-90:{$y_1'$}]{};
    \end{scope}
    \begin{scope}[xshift=4cm,yshift=-5*\yr cm]
     \node at (-0.5*\xr,0.5*\yr)[]{(c)};
     \node at (1.0*\xr,0.25*\yr)[anchor=west]{Graph~$S$};
     \theS{0}{0}
     \theScolor{}
     \theSpath{}
     \node at (xa)[label=90:{$x$}]{};
     \node at (ya)[label=90:{$y$}]{};
     \node at (ax6)[label=90:{$y_1$}]{};
     \node at (xb)[label=-90:{$x'$}]{};
     \node at (yb)[label=-90:{$y'$}]{};
     \node at (bx1)[label=0:{$x_1'$}]{};
     \node at (bx5)[label=180:{$y_2'$}]{};
     \node at (ab2)[label=0:{$z$}]{};
    \end{scope}

    \end{tikzpicture}
    \caption{The graphs 
    (a)~$L$,
    (b)~$H$,
    and (c)~$S$,
    For each graph, 
    an extendable partial (without~$y,y'$) 3-coloring is given.
    Moreover,
    for graph~$H$,
    an~$x$-$x'$ path and a~$y$-$y'$ path (blue) is depicted
    covering all vertices from the graph.
    For graph~$S$, 
    an~$x$-$y'$ path and an~$x'$-$y$ path (blue) is depicted
    covering all vertices from the graph.}
    \label{fig:gtb}
\end{figure}
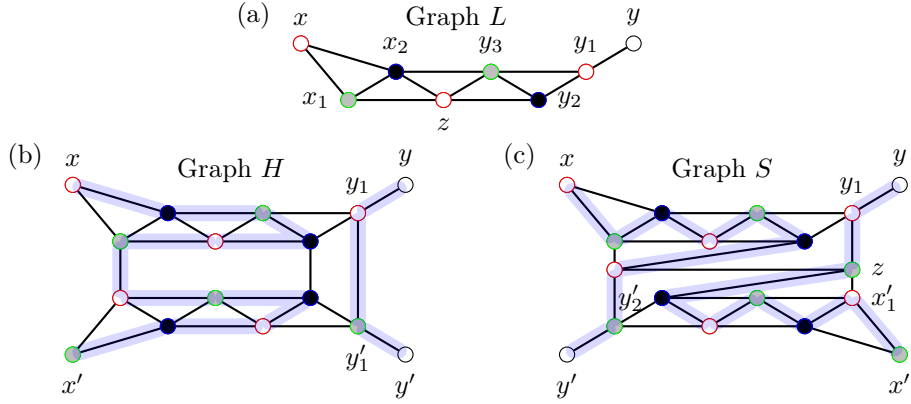
By this,
in each repetition,
we reduce the number of cycles in the 2-factor while preserving a 2-factor.
Finally,
we obtain a Hamiltonian cycle.

Each of the graphs~$H$ and~$S$,
contains two disjoint subgraphs isomorph to the graph~$L$
(see~\cref{fig:gtb}(a)).
We have the following important property of~$L$.

\begin{observation}
 \label{obs:L}
 Graph~$L$ is 3-colorable and for every valid 3-coloring~$f\colon V(L)\to\CS$,
 it holds that~$f(x)\neq f(y)$, 
 $f(x)=f(y_1)$,
 and~$f(x)\cup f(x_1)\cup f(y_2)=\CS$.
\end{observation}

\begin{proof}
 An extendable partial 3-coloring of~$L$ is depicted in~\cref{fig:gtb}(a).
 Let~$f\colon V(L)\to \CS$ be an arbitrary valid 3-coloring of~$L$.
 We have that~$f(x)\notin\{f(x_1),f(x_2)\}$,
 and hence~$f(z)=f(x)$.
 Since~$f(z)\notin \{f(y_2),f(y_3)\}$,
 we have that~$f(y_1)=f(x)$ and hence~$f(y)\neq f(x)$.
 Since~$f(x_2)\neq f(y_3)$,
 we have~$f(x_1)\neq f(y_2)$
 and thus~$f(x)\cup f(x_1)\cup f(y_2)=\CS$.
 \lqed
\end{proof}

We have the following important properties on~$H$ and~$S$.

\begin{lemma}
 \label{lem:HSHsSs}
 Graph~$X\in\{H,S\}$ is 3-colorable and for every valid 3-coloring~$f\colon V(X)\to\CS$,
 it holds that~$f(x)\neq f(y)$, $f(x')\neq f(y')$, and~$f(x)\neq f(x')$.
\end{lemma}

\begin{proof}
 We distinguish whether~$X=H$ or~$X=S$.
 
 \emph{Case 1: $X=H$}.
 An extendable partial 3-coloring of~$X$ is depicted \cref{fig:gtb}(b).
 Let~$f\colon V(X)\to\CS$ be a valid 3-coloring of~$X$.
 Due to~\cref{obs:L},
 we know that~$f(x)\neq f(y)$ and $f(x')\neq f(y')$.
 Moreover,
 we know that~$f(x)=f(y_1)$ and~$f(x')=f(y_1')$,
 and since~$\{y_1,y_1'\}\in E(X)$,
 it follows that~$f(x)\neq f(x')$.
 
 \emph{Case 2: $X=S$}.
 An extendable partial 3-coloring of~$X$ is depicted \cref{fig:gtb}(c).
 Let~$f\colon V(X)\to\CS$ be a valid 3-coloring of~$X$.
 Due to~\cref{obs:L},
 we know that~$f(x)\neq f(y)$ and $f(x')\neq f(y')$.
 Moreover,
 we know that~$f(x') \cup f(x_1')\cup f(y_2')=\CS$.
 Since~$\{y_1,x_1',y_2'\}\subseteq N_X(z)$,
 it follows that~$f(y_1)\neq f(x')$.
 Since~$f(x)=f(y_1)$,
 we thus have that~$f(x)\neq f(x')$.
 \lqed
\end{proof}

\begin{definition}[Insertion]
 Let~$G$ be a graph and
 $(v',v,u,u')\in V(G)^4$ be a quadruple with~$\{v',v\},\allowbreak\{v,u\},\{u,u'\}\in E(G)$.
 An~$X$-insertion 
 at~$(v',v,u,u')$ with~$X\in\{H,S\}$
 results in the graph obtained from~$G$
 by 
 deleting the edges~$\{v',v\}$, 
 $\{v,u\}$,
 and $\{u,u'\}$,
 adding a copy of~$X$ to~$G$, 
 and identifying~$y'$ with~$v'$,
 $x'$ with~$v$,
 $x$ with~$u$,
 and~$y$ with~$u'$.
\end{definition}

{
\begin{construction}
 \label{constr:2factormerging}
 Let~$G$ be a connected planar 4-regular graph, 
 and let~$\calQ$ be a 2-factor of~$G$ with~$|\calQ|>1$.
 Then there are two distinct cycles~$Q_i,Q_j\in \calQ$
 with two adjacent vertices~$u$ and~$v$ with~$u\in V(Q_i)$ and~$v\in V(Q_j)$.
 We distinguish two cases
 (see~\cref{fig:2factor} for an illustration).
 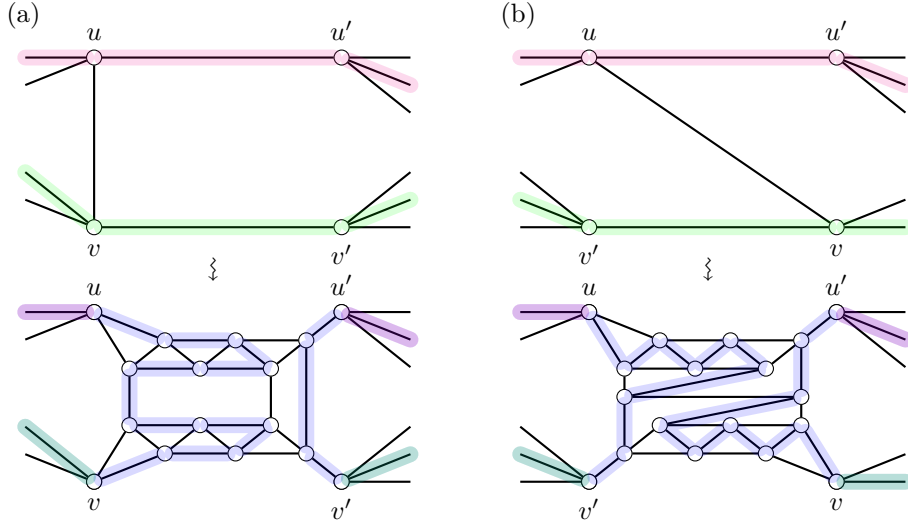
\begin{figure}[t]
  \centering
  \begin{tikzpicture}
   \def\xr{0.93}
   \def\yr{0.75}
   \def\xsh{1.66}
   \tikzpramble{}
   
   \newcommand{\exmex}{%
    \node (u) at (0,0)[xnode,label=90:{$u$}]{};
    \node (v) at (0*\xr,-3*\yr)[xnode,label=-90:{$v$}]{};
    \node (u1) at (-1*\xr,0*\yr)[inner sep=0pt]{};
    \node (u2) at (-1*\xr,-0.5*\yr)[inner sep=0pt]{};
    \node (u3) at (3.5*\xr,0*\yr)[xnode,label=90:{$u'$}]{};
    \node (u31) at (4.5*\xr,0*\yr)[inner sep=0pt]{};
    \node (u32) at (4.5*\xr,-0.5*\yr)[inner sep=0pt]{};
    \node (u33) at (4.5*\xr,-1*\yr)[inner sep=0pt]{};
    \node (v1) at (-1*\xr,-2.5*\yr)[inner sep=0pt]{};
    \node (v2) at (-1*\xr,-2*\yr)[inner sep=0pt]{};
    \node (v3) at (3.5*\xr,-3*\yr)[xnode,label=-90:{$v'$}]{};
    \node (v31) at (4.5*\xr,-3*\yr)[inner sep=0pt]{};
    \node (v32) at (4.5*\xr,-2.5*\yr)[inner sep=0pt]{};
    \node (v33) at (4.5*\xr,-2*\yr)[inner sep=0pt]{};
   }
   \newcommand{\exmaxA}{%
    \draw[xedge] (u) to (v);
    \foreach\x in {1,...,3}{\draw[xedge] (u) to (u\x);\draw[xedge] (v) to (v\x);}
    \foreach\x in {1,...,3}{\draw[xedge] (u3) to (u3\x);\draw[xedge] (v3) to (v3\x);}
   }
   \newcommand{\exmaxB}{%
    \foreach\x in {1,...,2}{\draw[xedge] (u) to (u\x);\draw[xedge] (v) to (v\x);}
    \foreach\x in {1,...,3}{\draw[xedge] (u3) to (u3\x);\draw[xedge] (v3) to (v3\x);}
   }
   \newcommand{\exmexx}{%
    \node (u) at (0,0)[xnode,label=90:{$u$}]{};
    \node (v) at (3.5*\xr,-3*\yr)[xnode,label=-90:{$v$}]{};
    \node (u1) at (-1*\xr,-0.5*\yr)[inner sep=0pt]{};
    \node (u2) at (-1*\xr,0*\yr)[inner sep=0pt]{};
    \node (u3) at (3.5*\xr,0*\yr)[xnode,label=90:{$u'$}]{};
    \node (u31) at (4.5*\xr,0*\yr)[inner sep=0pt]{};
    \node (u32) at (4.5*\xr,-0.5*\yr)[inner sep=0pt]{};
    \node (u33) at (4.5*\xr,-1*\yr)[inner sep=0pt]{};
    \node (v1) at (4.5*\xr,-2.5*\yr)[inner sep=0pt]{};
    \node (v2) at (4.5*\xr,-3*\yr)[inner sep=0pt]{};
    \node (v3) at (0*\xr,-3*\yr)[xnode,label=-90:{$v'$}]{};
    \node (v31) at (-1*\xr,-3*\yr)[inner sep=0pt]{};
    \node (v32) at (-1*\xr,-2.5*\yr)[inner sep=0pt]{};
    \node (v33) at (-1*\xr,-2*\yr)[inner sep=0pt]{};
   }
   \newcommand{\exmaxxA}{%
    \draw[xedge] (u) to (v);
    \foreach\x in {1,...,3}{\draw[xedge] (u) to (u\x);\draw[xedge] (v) to (v\x);}
    \foreach\x in {1,...,3}{\draw[xedge] (u3) to (u3\x);\draw[xedge] (v3) to (v3\x);}
   }
   \newcommand{\exmaxxB}{%
    \foreach\x in {1,...,2}{\draw[xedge] (u) to (u\x);\draw[xedge] (v) to (v\x);}
    \foreach\x in {1,...,3}{\draw[xedge] (u3) to (u3\x);\draw[xedge] (v3) to (v3\x);}
   }

   \begin{scope}[]
    \node at (-1*\xr,0.75*\yr)[]{(a)};
    \exmex{}
    \exmaxA{};
    \draw[xpathx] (u1) to (u) to (u3) to (u32);
    \draw[xpathy] (v2) to (v) to (v3) to (v32);
   \end{scope}
   
   \begin{scope}[yshift=-4.5*\yr cm]
    
    \node at (\xsh*\xr,0.75*\yr)[rotate=-90]{$\leadsto$};
    \exmex{}
    \exmaxB{};
    \begin{scope}%
    \theH{0}{0}
    \theHpath{}
    \end{scope}
    \draw[xpath] (u) to (u1);
    \draw[xpath] (v) to (v2);
    \draw[xpath] (u3) to (u32);
    \draw[xpath] (v3) to (v32);
    \draw[xpathx] (u) to (u1);
    \draw[xpathy] (v) to (v2);
    \draw[xpathx] (u3) to (u32);
    \draw[xpathy] (v3) to (v32);
   \end{scope}
   
   \begin{scope}[xshift=7.0*\xr cm]
    \node at (-1*\xr,0.75*\yr)[]{(b)};
    \exmexx{}
    \exmaxxA{}
    \draw[xpathx] (u2) to (u) to (u3) to (u32);
    \draw[xpathy] (v2) to (v) to (v3) to (v32);
   \end{scope}
   
   \begin{scope}[xshift=7.0*\xr cm,yshift=-4.5*\yr cm]
    \node at (\xsh*\xr,0.75*\yr)[rotate=-90]{$\leadsto$};
    \exmexx{}
    \exmaxxB{}
    \begin{scope}%
    \theS{0}{0}
    \theSpath{}
    \end{scope}
    
    \draw[xpath] (u2) to (u);
    \draw[xpath] (v) to (v2);
    \draw[xpath] (u3) to (u32);
    \draw[xpath] (v3) to (v32);
    \draw[xpathx] (u2) to (u);
    \draw[xpathy] (v) to (v2);
    \draw[xpathx] (u3) to (u32);
    \draw[xpathy] (v3) to (v32);
   \end{scope}
  \end{tikzpicture}
  \caption{Illustration to~\cref{constr:2factormerging} of (a) Case 1 and (b) Case 2.
  Indicated are
  two different 2-factor components in each upper part (magenta and green)
  and the ``merged'' cycle,
  being part of the new 2-factor with one less component,
  in the lower part (blue).}
  \label{fig:2factor}
 \end{figure}

 \begin{description}
  \item[Case 1:] There are two edges~$\{u',u\}\in E(Q_i)$ and~$\{v,v'\}\in E(Q_j)$
 sharing the same face 
 (see~\cref{fig:2factor}(a)).
 Make an~$H$-insertion at~$(v',v,u,u')$.
  \item[Case 2:] 
 There are no two edges~$e\in E(Q_i)$ and~$e'\in E(Q_j)$
 with~$u\in e$ and~$v\in e'$
 sharing the same face
 (see~\cref{fig:2factor}(b)).
 Since~$u$ and~$v$ are of degree four,
 there are two edges~$\{u,u'\}\in E(Q_i)$
 and
 $\{v,v'\}\in E(Q_j)$
 such that each share a face with~$e$.
 Make an~$S$-insertion at~$(v',v,u,u')$.
 \cqed
 \end{description}
\end{construction}
}

\cref{lem:HSHsSs} implies the following.

\begin{observation}
 \label{obs:corr}
  Let~$G'$ be the graph obtained from a graph~$G$ using~\cref{constr:2factormerging}.
  Then,
  $G$ is 3-colorable if and only if~$G'$ is 3-colorable.
\end{observation}

\citet[Theorems~2.1--2.3]{GareyJS76}
proved that \tcTsc{} is \NP-hard on \emph{connected} planar graphs of maximum vertex degree four.
\citet{Dailey80}
proved that
\tcTsc{} remains \NP-hard on connected 4-regular planar graphs.
We are set to prove the main result of this section.

\begin{proof}[Proof of~\cref{thm:tc:4regplanaHam}]
 Let~$(G)$ with~$G=(V,E)$ be an instance of \tcAcr{} 
 on connected 4-regular planar graphs.
 Compute a 2-factor~$\calQ$ of~$G$ in polynomial time~\cite{Mulder92,FleischnerS03}.
 If~$\calQ=\{Q_1\}$,
 then return~$(G,Q_1)$.
 Otherwise,
 apply~\cref{constr:2factormerging} iteratively.
 In each iteration~$i\geq 1$,
 we obtain an equivalent
 instance (\cref{obs:corr})
 and a 2-factor,
 say~$\calQ'$, with~$|\calQ'|=|\calQ|-i$.
 Thus,
 after at most~$n/3$ steps~\cite{FleischnerS03},
 we obtain a graph~$G'$ and a Hamiltonian cycle~$\HC$ of~$G'$.
 Finally,
 return the instance~$(G',\HC)$.
 \lqed
\end{proof}

In the next section,
we prove that \tcTsc{} is \NP-hard on 5-regular planar Hamiltonian graphs.
For this,
we need an even input graph as we will connect disjoint pairs of vertices with some gadget.
Building on \cref{thm:tc:4regplanaHam},
we prove that we can further turn a 4-regular planar Hamiltonian graph into
an even 4-regular planar Hamiltonian graph.

\begin{proposition}
 \label{prop:4even}
 \tcTsc{} is \NP-hard on even 4-regular planar Hamiltonian graphs,
 even if a Hamiltonian cycle is given.
\end{proposition}

To prove~\cref{prop:4even},
we will add a disjoint copy of the input graph to itself and
connect the two graphs using two copies of the graph~$W$
(see~\cref{fig:geteven}(a)).
The graph~$W$ is inspired by a graph of~\citet[Fig.~3]{Dailey80}.
It is not difficult to see that following holds true for graph~$W$.

\begin{observation}
 \label{obs:W}
 Graph~$W$ is 3-colorable, 
 for every valid 3-coloring~$f\colon V(W)\to\CS$,
 it holds true that~$f(x)=f(x_0)=f(y_0)=f(y)$,
 and for every~$c\in\CS$,
 there is a valid 3-coloring with~$f(x)=c$.
\end{observation}

\begin{construction}
 \label{constr:geteven}
 Let~$(G,C)$ be an instance of~\tcAcr{}
 on 4-regular planar Hamiltonian graphs
 with~graph~$G=(V,E)$ and Hamiltonian cycle~$\HC$.
 Construct a graph~$G^*$ and a Hamiltonian cycle~$H^*$ of~$G^*$ as follows
 (see~\cref{fig:geteven} for an illustration).
 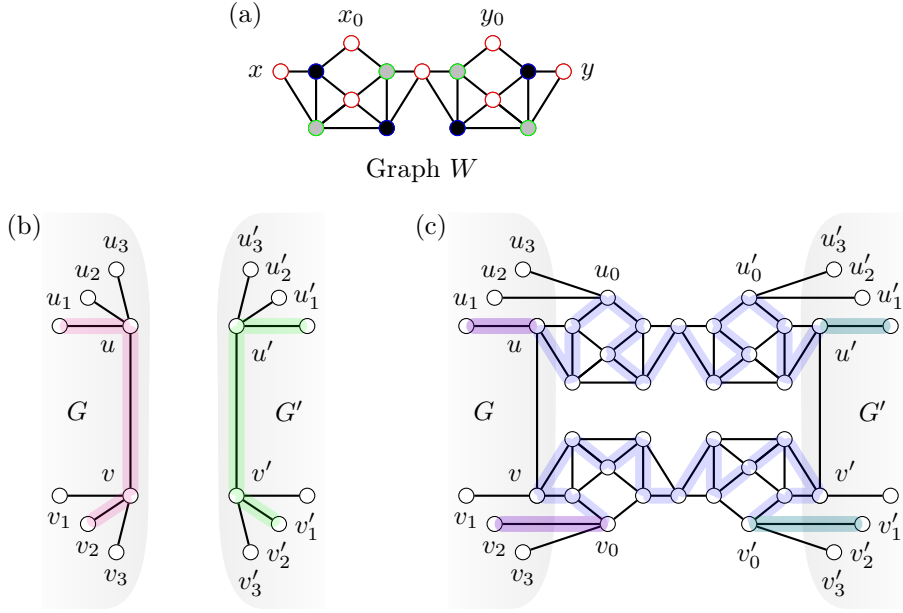
\begin{figure}[t]
  \centering
  \begin{tikzpicture}
   \def\xr{0.93}
   \def\yr{0.75}
   \def\xsh{1.66}
   \tikzpramble{}
   
   \newcommand{\exmex}[4]{%
    \node (#1u) at (0,0)[xnode,label=-#4:{\ensuremath{#2}}]{};
    \node (#1u1) at (-1*\xr,0*\yr)[xnode,label={\ensuremath{#2_1}}]{};
    \node (#1u2) at (-0.6*\xr,0.5*\yr)[xnode,label={\ensuremath{#2_2}}]{};
    \node (#1u3) at (-0.2*\xr,1*\yr)[xnode,label={\ensuremath{#2_3}}]{};
    
    \node (#1v) at (0*\xr,-3*\yr)[xnode,label=#4:{\ensuremath{#3}}]{};
    \node (#1v1) at (-1*\xr,-3*\yr)[xnode,label=-90:{\ensuremath{#3_1}}]{};
    \node (#1v2) at (-0.6*\xr,-3.5*\yr)[xnode,label=-90:{\ensuremath{#3_2}}]{};
    \node (#1v3) at (-0.2*\xr,-4*\yr)[xnode,label=-90:{\ensuremath{#3_3}}]{};
   }
   \newcommand{\exmaxA}[1]{%
    \draw[xedge] (#1u) to (#1v);
    \foreach\x in {1,...,3}{\draw[xedge] (#1u) to (#1u\x);\draw[xedge] (#1v) to (#1v\x);}
   }
   \newcommand{\theDX}[5]{%
    \node (#1x0) at (0,0) [xnode]{};
    \node (#1a1) at (-0.5*\xr,0.0*\yr)[xnode]{};
    \node (#1a2) at (-0.5*\xr,-1*\yr)[xnode]{};
    \node (#1a3) at (-1*\xr,-0.5*\yr)[xnode]{};
    \node (#1a4) at (-1.5*\xr,0*\yr)[xnode]{};
    \node (#1a5) at (-1.5*\xr,-1*\yr)[xnode]{};
    \node (#1aA) at (-2*\xr,0*\yr)[xnode,label=#2]{};
    \node (#1aB) at (-1*\xr,0.5*\yr)[xnode,label=#3]{};
    \foreach\x/\y in {#1aA/#1a4,#1aA/#1a5,#1aB/#1a4,#1aB/#1a1,#1a4/#1a5,#1a3/#1a5,#1a1/#1a2,#1a2/#1a5}{\draw[xedge] (\x) to (\y);}
    \foreach\x in {#1a4,#1a5,#1a2,#1a1}{\draw[xedge](#1a3) to (\x);}
    \node (#1b1) at (0.5*\xr,0*\yr)[xnode]{};
    \node (#1b2) at (0.5*\xr,-1*\yr)[xnode]{};
    \node (#1b3) at (1*\xr,-0.5*\yr)[xnode]{};
    \node (#1b4) at (1.5*\xr,0*\yr)[xnode]{};
    \node (#1b5) at (1.5*\xr,-1*\yr)[xnode]{};
    \node (#1bA) at (2*\xr,0*\yr)[xnode,label=#4]{};
    \node (#1bB) at (1*\xr,0.5*\yr)[xnode,label=#5]{};
    \foreach\x/\y in {#1bA/#1b4,#1bA/#1b5,#1bB/#1b4,#1bB/#1b1,#1b4/#1b5,#1b3/#1b5,#1b1/#1b2,#1b2/#1b5}{\draw[xedge] (\x) to (\y);}
    \foreach\x in {#1b4,#1b5,#1b2,#1b1}{\draw[xedge](#1b3) to (\x);}
    \foreach\x in {#1a1,#1a2,#1b2,#1b1}{\draw[xedge](#1x0) to (\x);}
   }

   \begin{scope}[xshift=4.375*\xr cm,yshift=4.5*\yr cm]
        \node at (-2.5*\xr,1*\yr)[]{(a)};
        \theDX{X}{180:$x$}{90:$x_0$}{0:$y$}{90:$y_0$};
        \node at (0*\xr,-1.75*\yr)[]{Graph~$W$};
        \foreach\x\y in {XaA/xnodeA,XaB/xnodeA,Xa3/xnodeA,Xx0/xnodeA,XbA/xnodeA,Xb3/xnodeA,XbB/xnodeA}
        {\node at (\x)[\y]{};}
        \foreach\x\y in {Xa1/xnodeB,Xa2/xnodeC,Xa4/xnodeC,Xa5/xnodeB}
        {\node at (\x)[\y]{};}
        \foreach\x\y in {Xb1/xnodeB,Xb2/xnodeC,Xb4/xnodeC,Xb5/xnodeB}
        {\node at (\x)[\y]{};}
   \end{scope}
   
   \newcommand{\Gshape}[6]{%
      \shade[left color=#5, right color=#6,rounded corners,draw=none] (-1.25*\xr,2*\yr) to (-0.5,2*\yr) to [out=#1,in=#2](0.25,0.*\yr) to (0.25*\xr,-3*\yr) to [out=#3,in=#4](-0.5,-5.*\yr) to (-1.25*\xr,-5*\yr);
   }
   \def\XgrayD{lightgray!25!white}
   \def\XgrayB{lightgray!5!white}

   \begin{scope}[xshift=0.25*\xr cm]
    \Gshape{0}{90}{-90}{0}{\XgrayB}{\XgrayD};
    \node at (-1.5*\xr,1.75*\yr)[]{(b)};
    \exmex{A}{u}{v}{135}
    \exmaxA{A};
    \draw[xpathx] (Au1) to (Au) to (Av) to (Av2);
    \node at (-0.75*\xr, -1.5*\yr)[]{$G$};
   \end{scope}
   
   \begin{scope}[xshift=1.75*\xr cm,x=-1 cm]
   \Gshape{180}{90}{-90}{180}{\XgrayD}{\XgrayB};
    \exmex{B}{u'}{v'}{45}
    \exmaxA{B};
    \draw[xpathy] (Bu1) to (Bu) to (Bv) to (Bv2);
    \node at (-0.75*\xr, -1.5*\yr)[]{$G'$};
   \end{scope}
   
   \begin{scope}[xshift=6*\xr cm]
    \Gshape{0}{90}{-90}{0}{\XgrayB}{\XgrayD};
    \node at (-1.5*\xr,1.75*\yr)[]{(c)};
    \exmex{A}{u}{v}{135}
    \node at (-0.75*\xr, -1.5*\yr)[]{$G$};
   \end{scope}
   
   \begin{scope}[xshift=10*\xr cm,x=-1 cm]
    \Gshape{180}{90}{-90}{180}{\XgrayD}{\XgrayB};
    \exmex{B}{u'}{v'}{45}
    \node at (-0.75*\xr, -1.5*\yr)[]{$G'$};
   \end{scope}
   
   \begin{scope}[xshift=8*\xr cm]
    \theDX{A}{}{90:$u_0$}{}{90:$u_0'$};
   \end{scope}
   \begin{scope}[xshift=8*\xr cm,yshift=-3*\yr cm,y=-1 cm]
    \theDX{B}{}{-90:$v_0$}{}{-90:$v_0'$};
   \end{scope}

   \foreach\x in{Au2,Au3}{\draw[xedge] (AaB) to (\x);}
   \foreach\x in{Av2,Av3}{\draw[xedge] (BaB) to (\x);}
   \draw[xedge] (AaA) to (Au1);
   \draw[xedge] (BaA) to (Av1);
   \draw[xedge] (AaA) to (BaA);
   \foreach\x in{Bu2,Bu3}{\draw[xedge] (AbB) to (\x);}
   \foreach\x in{Bv2,Bv3}{\draw[xedge] (BbB) to (\x);}
   \draw[xedge] (AbA) to (Bu1);
   \draw[xedge] (BbA) to (Bv1);
   \draw[xedge] (AbA) to (BbA);
   
   \draw[xpathx] (Av2) to (BaB);
   \draw[xpathx] (Au1) to (AaA);
   \draw[xpathy] (Bu1) to (AbA);
   \draw[xpathy] (Bv2) to (BbB);
   
   \draw[xpath] (Av2) to (BaB) to (Ba4) to (BaA) to (Ba5) to (Ba3) to (Ba2) to (Ba1) to  (Bx0) to (Bb2) to (Bb1) to (Bb3) to (Bb5) to (BbA) to (Bb4) to (BbB) to (Bv2);
   \draw[xpath] (Bu1) to (AbA) to (Ab5) to (Ab4) to (AbB) to (Ab1) to (Ab3) to (Ab2) to (Ax0) to (Aa2) to (Aa3) to (Aa1) to (AaB) to (Aa4) to (Aa5) to (AaA) to (Au1); 
  
  \end{tikzpicture}
  \caption{Illustration to the proof of \cref{prop:4even}.
  (a) the graph~$W$ with a valid 3-coloring,
  (b) excerpts of the graphs~$G$ and~$G'$ with Hamiltonian cycle~$\HC$ (magenta)
  and~$\HC'$ (green), and
  (c) excerpt of the graph~$G^*$ with Hamiltonian cycle~$\HC^*$ (blue).}
  \label{fig:geteven}
 \end{figure}
 Take~$G$ and a disjoint copy~$G'$ of~$G$.
 Let~$e=\{u,v\}$ be an edge in~$\HC$ of~$G$,
 and let~$e'=\{u',v'\}$ its copy in~$G'$.
 Let~$G$ be embedded such that
 such that~$e$ is incident with the outer face~\cite{balakrishnan2012textbook}.
 Let~$G'$ have the same embedding as~$G$ but mirrored along the $y$-axis
 (see~\cref{fig:geteven}(b); we assume such an embedding from now on).
 Connect the graphs~$G$ and~$G'$ through the vertices~$u,u',v,v'$ as follows
 (see~\cref{fig:geteven}(c)).
 
 Add a copy of the graph~$W$,
 and identify~$x$ with~$u$ (call the vertex again~$u$)
 and~$y$ with~$u'$ (call the vertex again~$u'$)
 and rename~$x_0$ by~$u_0$ and~$y_0$ by~$u_0'$.
 Let~$v,u_1,u_2,u_3$ be the neighbors of~$u$ in clockwise order.
 Remove the edges~$\{u,u_2\}$ and~$\{u,u_3\}$,
 and add the edges~$\{u_0,u_2\}$ and~$\{u_0,u_3\}$.
 Let~$v',u_1',u_2',u_3'$ be the neighbors of~$u'$ in counter-clockwise order.
 Remove the edges~$\{u',u_2'\}$ and~$\{u',u_3'\}$,
 and add the edges~$\{u_0',u_2'\}$ and~$\{u_0',u_3'\}$.
 
 Add another copy of the graph~$W$,
 and identify~$y$ with~$v$ (call the vertex again~$v$)
 and~$x$ with~$v'$ (call the vertex again~$v'$)
 and rename~$y_0$ by~$v_0$ and~$x_0$ by~$v_0'$.
 Let~$u,v_1,v_2,v_3$ be the neighbors of~$v$ in counter-clockwise order.
 Remove the edges~$\{v,v_2\}$ and~$\{v,v_3\}$,
 and add the edges~$\{v_0,v_2\}$ and~$\{v_0,v_3\}$.
 Let~$u',v_1',v_2',v_3'$ be the neighbors of~$v'$ in clockwise order.
 Remove the edges~$\{v',v_2'\}$ and~$\{v',v_3'\}$,
 and add the edges~$\{v_0',v_2'\}$ and~$\{v_0',v_3'\}$.
 
 We can merge~$\HC$ and~$\HC'$ through the two added~$W$'s as depicted in 
 \cref{fig:geteven}(c)
 to a Hamiltonian cycle~$\HC^*$ of~$G^*$.
 \cqed
\end{construction}

\begin{proof}[Proof of~\cref{prop:4even}]
 Let~$(G,C)$ be an instance of the \NP-hard (\cref{thm:tc:4regplanaHam})
 \tcAcr{}
 on 4-regular planar Hamiltonian graphs
 with graph~$G=(V,E)$ and Hamiltonian cycle~$\HC$.
 If~$G$ is even,
 then we return~$(G,\HC)$.
 Assume that~$G$ is odd.
 Construct the graph~$G^*$ with Hamilton cycle~$H^*$ from~$G$ using~\cref{constr:geteven}.
 Clearly,
 $G^*$ is Hamiltonian,
 4-regular, 
 planar,
 and even (note that~$|V(G^*)|=2|V(G)|+26$).
 Finally,
 due to~\cref{obs:W},
 we have that~$G$ is 3-colorable if and only if
 $G^*$ is 3-colorable.
 \lqed
\end{proof}

\subsection{5-regular planar Hamiltonian}
\label{ssec:5regplanarham}

In this section,
we prove that~\tcTsc{}
is also \NP-hard on 5-regular planar Hamiltonian graphs with provided Hamiltonian cycle.

\begin{theorem}%
 \label{thm:tc:5regplanham}
 \tcTsc{} on 5-regular planar Hamiltonian graphs is \NP-hard,
 even if a Hamiltonian cycle is given.
\end{theorem}

\noindent
We will use multiple copies of the graph~$D$
(see~\cref{fig:dd})
to turn a planar 4-regular Hamiltonian graph into
a planar 5-regular Hamiltonian graph.
\begin{figure}[t]
 \centering
  \begin{tikzpicture}
    \def\xr{1.33}
    \def\yr{1.25}
    \tikzpramble{};
     \begin{scope}[xshift=0*\xr cm,yshift=0*\yr cm]
     \theDD{0}{0}
    \end{scope}
    
    \node at (Ax)[label=180:$x$]{};
    \node at (Bx)[label=180:$y$]{};
    
    \theDDpath{}
    \draw[xpath] (Ax) to (Aa1);
    \draw[xpath] (Bx) to (Ba1);
    
    \newcommand{\theDDcolor}[4]{%
      \foreach\x\y in{#1x/#2,#1a1/#3,#1b2/#2,#1a2/#4,#1a3/#3,#1b3/#4,#1a4/#2,#1b4/#2,#1c2/#3,#1b1/#3,#1c3/#4,#1b5/#3,#1b6/#2,#1a5/#4,#1a6/#3,#1b7/#3,#1c4/#4,#1b8/#2,#1c5/#3,#1a7/#4,#1b9/#3,#1d1/#2,#1c1/#4,#1b10/#2,#1a9/#4,#1a8/#2}
     {\node at (\x)[\y]{};}
    }
    \theDDcolor{A}{xnodeA}{xnodeB}{xnodeC}
    \theDDcolor{B}{xnodeC}{xnodeA}{xnodeB}
    \end{tikzpicture}
    \caption{The graph~$D$.
    A valid 3-coloring and an~$x$-$y$ Hamiltonian path (blue)
    are depicted.}
    \label{fig:dd}
\end{figure}
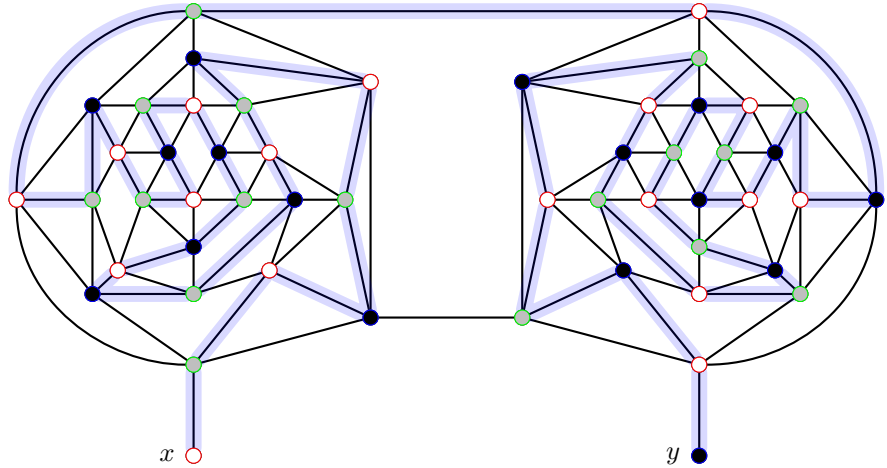
The graph~$D$ is inspired by a graph of~\citet[Fig.~4]{Dailey80}.
We have the following.

\begin{observation}
 \label{obs:D}
 Graph~$D$ is 3-colorable
 and
 for every distinct~$c,c'\in\CS$,
 there is a valid 3-coloring~$f\colon V(D)\to \CS$ with~$f(x)=c$ and~$f(y)=c'$.
\end{observation}

We will perform a series of~$D$-insertions defined as follows.

\begin{definition}[$D$-insertion]
 Let~$G$ be a graph and
 $\{v,u\}\in E(G)$.
 A~$D$-insertion 
 at~$\{v,u\}$
 results in the graph obtained from~$G$
 by 
 adding a copy of~$D$ to~$G$, 
 and identifying~$x$ with~$v$
 and~$y$ with~$u$.
\end{definition}

\begin{proof}[Proof of~\cref{thm:tc:5regplanham}]
 Let~$(G,\HC)$ be an instance 
 of the \NP-hard (\cref{prop:4even})
 \tcAcr{}
 on even 4-regular planar Hamiltonian graphs
 with graph~$G=(V,E)$
 and Hamiltonian cycle~$\HC=(v_1,\dots,v_n)$ of~$G$.
 Construct a graph~$G'$ with Hamiltonian cycle~$\HC'$ as follows
 (see~\cref{fig:dinsertion} for an illustration).
\begin{figure}[t]
  \centering
  \begin{tikzpicture}
   \def\xr{1}
   \def\yr{1}
   \tikzpramble{}
   
   \def\exO{2}
   \def\eyO{0.5}
   \begin{scope}[xshift=0cm,yshift=0cm]
    \theEllipse{\exO}{\eyO}{}
    \draw[xpathx] (cd) to (b2) to (a5) to (a4) to (a3) to (a2) to (a1) to (a0) to (b1) to (cd);
    \node at (0,0)[]{$G$};
    \node at (3.5*\xr,0*\yr)[]{$\leadsto$};
   \end{scope}
   
   \begin{scope}[xshift=7*\xr cm,yshift=0cm]
    \theEllipse{\exO}{\eyO}{}
    \node at (0,0)[]{$G'$};
   
   \begin{scope}[xshift=-1.9*\xr cm,yshift=1.4*\yr cm]
     \begin{scope}[scale=0.25,transform shape,rotate=20]
     \theDDx{0}{0}
     \theDDpath{A}{180}
     \theDDpath{B}{0}
     \draw[xedge] (a4) to (Aa1);
     \draw[xedge] (a3) to (Ba1);
     \draw[xpath] (a4) to (Aa1);
     \draw[xpath] (a3) to (Ba1);
     \end{scope}
   \end{scope}
   
   \begin{scope}[xshift=-1.45*\xr cm,yshift=-1.65*\yr cm,scale=0.25,transform shape,rotate=140]
     \theDDx{0}{0}
     \theDDpath{A}{180}
     \theDDpath{B}{0}
     \draw[xedge] (b2) to (Aa1);
     \draw[xedge] (a5) to (Ba1);
     \draw[xpath] (b2) to (Aa1);
     \draw[xpath] (a5) to (Ba1);
   \end{scope}

   \begin{scope}[xshift=0.8*\xr cm,yshift=1.7*\yr cm,scale=0.25,transform shape,rotate=340]
     \theDDx{0}{0}
     \theDDpath{A}{180}
     \theDDpath{B}{0}
     \draw[xedge] (a2) to (Aa1);
     \draw[xedge] (a1) to (Ba1);
     \draw[xpath] (a2) to (Aa1);
     \draw[xpath] (a1) to (Ba1);
   \end{scope}
   
   \begin{scope}[xshift=2.25*\xr cm,yshift=-1.*\yr cm,scale=0.25,transform shape,rotate=220]
     \theDDx{0}{0}
     \theDDpath{A}{180}
     \theDDpath{B}{0}
     \draw[xedge] (a0) to (Aa1);
     \draw[xedge] (b1) to (Ba1);
     \draw[xpath] (a0) to (Aa1);
     \draw[xpath] (b1) to (Ba1);
   \end{scope}
   
    \draw[xpath] (cd) to (b2);
    \draw[xpath] (a5) to (a4);
    \draw[xpath] (a3) to (a2);
    \draw[xpath] (a1) to (a0);
    \draw[xpath] (b1) to (cd);
   \end{scope}

  \end{tikzpicture}
  \caption{Illustration to the proof of~\cref{thm:tc:5regplanham}.
  The magenta path depicts the Hamiltonian cycle before the $D$-insertions,
  and the blue path depicts the Hamiltonian cycle after the~$D$-insertions.}
  \label{fig:dinsertion}
 \end{figure}
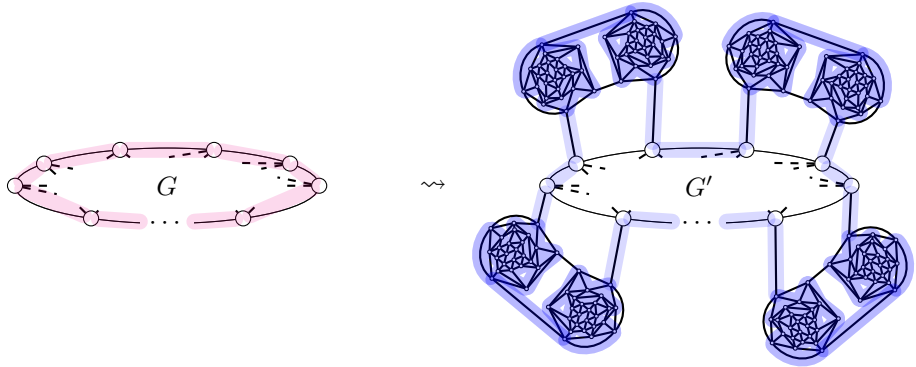
 For each~$i\in\set{n/2}$,
 make a~$D$-insertion at~$\{v_{2i-1},v_{2i}\}$.
 For~$\HC'$,
 replace each edge~$\{v_{2i-1},v_{2i}\}$ by the $x$-$y$ Hamiltonian path through the~$D$ inserted at~$\{v_{2i-1},v_{2i}\}$.
 Note that~$G'$ is 5-regular and planar.
 Finally,
 due to~\cref{obs:D},
 $G'$ is 3-colorable if and only if~$G$ is 3-colorable.
 \lqed
\end{proof}

\section{Regular Hamiltonian Graphs}
\label{sec:reg}

In this section,
we prove that \tcTsc{} remains \NP-hard
on $p$-regular Hamiltonian graphs for \emph{every}~$p\in\N_{\geq 4}$.

\begin{theorem}%
 \label{thm:tc:kreg}
 For every~$p\in\N_{\geq 4}$,
 \tcTsc{} on $p$-regular Hamiltonian graphs is \NP-hard,
 even if a Hamiltonian cycle is provided.
\end{theorem}

\begin{proof}
 We prove the statement via induction.
 Due to~\cref{thm:tc:5regplanham,thm:tc:4regplanaHam},
 we know that \tcAcr{} is \NP-hard on $p$-regular Hamiltonian graphs for~$p\in\{4,5\}$.
 Assume that the statement is true for~$p'\in\set[3,4]{p}$,
 $p\in\N_{\geq 4}$.
 Let~$(G,\HC)$ be an instance of
 the \NP-hard (by induction)
 \tcAcr{} on $p$-regular Hamiltonian graphs 
 with graph~$G=(V,E)$ and
 Hamiltonian cycle~$\HC=(v_0,\dots,v_{n-1})$ of~$G$,
 where~$n=|V|$.
 We construct a graph~$G^*$ from~$G$ as follows
 (see~\cref{fig:stacked}).
 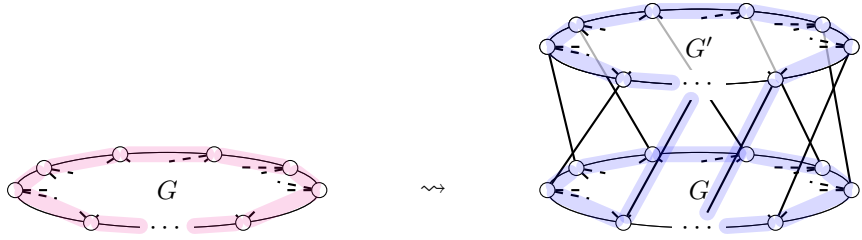
\begin{figure}[t]
  \centering
  \begin{tikzpicture}
    \def\xr{1}
    \def\yr{1}
    \tikzpramble{}
    \begin{scope}[xshift=-7*\xr cm]
      \theEllipse{2}{0.5}{A}
      \draw[xpathx] (Acd) to (Ab2) to (Aa5) to (Aa4) to (Aa3) to (Aa2) to (Aa1) to (Aa0) to (Ab1) to (Acd);
      \node at (0,0)[]{$G$};
      \node at (3.5*\xr,0*\yr)[]{$\leadsto$};
    \end{scope}
    \theEllipse{2}{0.5}{A}
    \node at (0,0)[]{$G$};
    \begin{scope}[yshift=1.9*\yr cm]
      \theEllipse{2}{0.5}{B}
    \end{scope}
    \foreach\x\y in {Ab2/Bcd,Acd/Bb1,Ab1/Ba0,Aa0/Ba1,Aa1/Ba2,Aa2/Ba3,Aa3/Ba4,Aa4/Ba5,Aa5/Bb2}{\draw[xedge] (\x) to (\y);}
    \begin{scope}[yshift=1.9*\yr cm]
      \theEllipse{2}{0.5}{B}
      \node at (0,0)[]{$G'$};
    \end{scope}
    \draw[xpath] (Ab2) to (Aa5) to (Aa4) to (Aa3) to (Aa2) to (Aa1) to (Aa0) to (Ab1) to (Acd);
    \draw[xpath] (Bcd) to (Bb2) to (Ba5) to (Ba4) to (Ba3) to (Ba2) to (Ba1) to (Ba0) to (Bb1);
    \draw[xpath] (Bcd) to (Ab2);
    \draw[xpath] (Acd) to (Bb1);
  \end{tikzpicture}
  \caption{Illustration to the proof of~\cref{thm:tc:kreg}.
  The input graph~$G$ (left-hand side) and the output graph~$G^*$ (right-hand side).
  For~$G$ and~$G^*$ are Hamiltonian cycles depicted in magenta and blue,
  respectively.}
  \label{fig:stacked}
  \end{figure}
 Add a disjoint copy~$G'=(V',E')$ of~$G$ to~$G$.
 Denote the vertices of~$G'$ by~$v_0',\dots,v_{n-1}'$,
 and the copy of~$\HC$ in~$G'$ by~$\HC'$.
 Finally,
 add the edge set~$\{\{v_{i},v_{i+1\bmod n}'\}\mid i\in\set[0]{n-1}\}$.
 Observe that~$G^*$ admits the following Hamiltonian cycle~$\HC^*=(v_{n-1},v_0',v_{n-1}',\dots,v_1',v_0,v_1,\dots,v_{n-2},v_{n-1})$.
 Clearly,~$G^*$ is~$(p+1)$-regular.
 Moreover,
 since~$G\subseteq G^*$, 
 every valid 3-coloring of~$G^*$ induces a valid 3-coloring of~$G$. 
 We claim that if~$G$ is 3-colorable, 
 then~$G^*$ is 3-colorable.
 
 Let~$f\colon V(G)\to\CS$ be a valid 3-coloring of~$G$.
 We claim that~$f^*\colon V(G^*)\to \CS$ 
 with~$f^*(v_i) = f(v_i)$ and~$f^*(v_i')=f(v_i)$
 for all~$i\in\set[0]{n-1}$ is a valid 3-coloring of~$G^*$.
 Clearly,
 for every~$Z\in\{V,V'\}$
 and for every~$\{v,w\}\in E(G^*)\cap \binom{Z}{2}$,
 we have~$f^*(v)\neq f^*(w)$.
 Moreover,
 for each~$i\in\set[0]{n-1}$,
 we have that~$f^*(v_i)\neq f^*(v_{(i+1)\bmod n}) = f^*(v_{(i+1)\bmod n}')$,
 and thus each edge~$\{v_i,v_{(i+1)\bmod n}'\}$ has differently colored endpoints.
 It follows that~$f^*$ is a valid 3-coloring of~$G^*$.
 \lqed
\end{proof}

\section{Ordered and Connected Hamiltonian Graphs}
\label{sec:ordcon}

In this section,
we prove~\tcTsc{} to remain \NP-hard on $p$-ordered Hamiltonian graphs
for every~$p\geq 3$.

\begin{theorem}
 \label{thm:tc:pordham}
 \tcTsc{} is \NP-hard on $p$-ordered regular Hamiltonian graphs
 for every~$p\in\N_{\geq 3}$,
 even if a Hamiltonian cycle is given.
\end{theorem}

Every Hamiltonian graph is 3-Hamiltonian-ordered 
and hence 3-ordered Hamiltonian.
Starting from here,
we will inductively prove~\cref{thm:tc:pordham}.
In each inductive step,
we will construct a graph of connectivity high enough and employ the following.

\begin{fact}[\cite{BollobasT96,FaudreeF02}]
 \label{fact:conn2ord}
 If~$G$ is~$22p$-connected graph for some~$p\in\N_{\geq 3}$,
 then~$G$ is $p$-ordered.
\end{fact}

\begin{construction}[$q$-complete graph]
 \label{constr:qcompl}
 Let~$G=(V,E)$ and let~$q\in\N$.
 Construct the \emph{$q$-complete} graph~$C_q(G)=(V_q,E_q)$ as follows
 (see~\cref{fig:complete} for an illustration).
 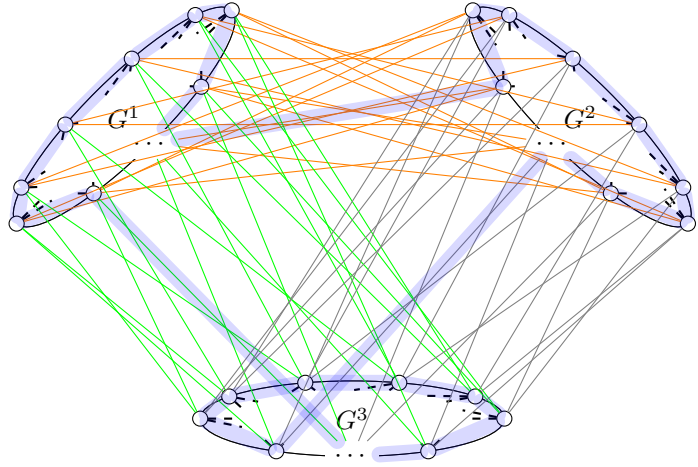
\begin{figure}
  \centering
  \begin{tikzpicture}
    \def\xr{1}
    \def\yr{1}
    \tikzpramble{}
    
    \begin{scope}[rotate=0]
    \begin{scope}[xshift=-3*\yr cm,yshift=5*\yr cm,rotate=45]
      \theEllipse{2}{0.5}{A}
      \node at (0,0)[]{$G^1$};
    \end{scope}
    \begin{scope}[xshift=3*\yr cm,yshift=5*\yr cm,rotate=-45]
      \theEllipse{2}{0.5}{B}
      \node at (0,0)[]{$G^2$};
    \end{scope}
    
    \newcommand{\supcon}[3]{
      \foreach\x\y in {#1b2/#2cd,#1cd/#2b1,#1b1/#2a0,#1a0/#2a1,#1a1/#2a2,#1a2/#2a3,#1a3/#2a4,#1a4/#2a5,#1a5/#2b2,
      #2b2/#1cd,#2cd/#1b1,#2b1/#1a0,#2a0/#1a1,#2a1/#1a2,#2a2/#1a3,#2a3/#1a4,#2a4/#1a5,#2a5/#1b2}{\draw[#3] (\x) to (\y);}
    }
    \begin{scope}[yshift=1*\yr cm]
      \theEllipse{2}{0.5}{C}
      \node at (0,0)[]{$G^3$};
    \end{scope}
    
    \supcon{A}{B}{orange}
    \supcon{A}{C}{green}
    \supcon{B}{C}{gray}
    
    \newcommand{\aRound}[1]{
    \draw[xpath] (#1b2) to (#1a5) to (#1a4) to (#1a3) to (#1a2) to (#1a1) to (#1a0) to (#1b1) to (#1cd);
    }
    \aRound{A}
    \aRound{B}
    \aRound{C}
    \draw[xpath] (Acd) to (Bb2);
    \draw[xpath] (Bcd) to (Cb2);
    \draw[xpath] (Ccd) to (Ab2);
    \end{scope}
  \end{tikzpicture}
  \caption{Illustration of~$C_q(G)$ with~$q=3$. $G^1,G^2,G^3$ denote the three copies of~$G$, 
  illustrated through a Hamiltonian cycle of~$G$. 
  Orange edges connect vertices from~$G^1$ and~$G^2$,
  green edges connect vertices from~$G^1$ and~$G^3$,
  and gray edges connect vertices from~$G^2$ and~$G^3$.
  Only adjacencies along the Hamiltonian cycle are depicted
  (not all edges are shown).
  Moreover,
  a Hamiltonian cycle in~$C_3(G)$ is depicted (blue).}
  \label{fig:complete}
  \end{figure}
 Let~$V_q\ceq V^1\uplus \dots\uplus V^q$
 where~$V^i\ceq \{v^i\mid v\in V\}$ for all~$i\in\set{q}$.
 Let~$E_q\ceq \bigcup_{i,j\in\set{q}}\{\{u^i,v^j\}\mid \{u,v\}\in E\}$.
 \cqed
\end{construction}

By construction,
we have the following.

\begin{observation}
 \label{obs:qcompl:3col}
 Let~$q\in\N$.
 Then,
 $G$ is 3-colorable if and only if~$C_q(G)$ is 3-colorable.
\end{observation}

Moreover,
the~$q$-complete graph preserves
Hamiltonicity.

\begin{lemma}
 \label{lem:compl:ham}
 Let~$q\in\N$.
 If~$G$ admits a Hamiltonian cycle,
 then
 $C_q(G)$ admits a Hamiltonian cycle computable in polynomial time.
\end{lemma}

\begin{proof}
 Let~$\HC=(v_0,\dots,v_{n-1},v_0)$ be a Hamiltonian cycle of~$G$.
 Then
 \[\HC_q\ceq (v_0^1,\dots,v_{n-1}^1,v_0^2,\dots,v_{n-1}^2,v_0^3,\dots,\dots,v_{n-1}^{q},v_0^1)\]
 is a Hamiltonian cycle in~$C_q(G)$.
 \lqed
\end{proof}

Now we argue about the connectivity.

\begin{lemma}
 \label{lem:tc:qconn}
 Let~$q\in\N$.
 If~$G$ is~$p$-ordered for some~$p\in\N_{\geq 3}$,
 then~$C_q(G)$ is~$q(p-1)$-connected.
\end{lemma}

\begin{proof}
 Let~$G$ be~$p$-ordered for some~$p\in\N_{\geq 3}$.
 Let~$v,w\in V$ be two arbitrary vertices of~$G$
 and let~$x,y\in\set{q}$ be arbitrary.
 Since~$G$ is~$p$-ordered,
 $G$ is~$(p-1)$-connected~\cite{NgS97},
 and thus there are~$p-1$ mutually internally vertex-disjoint $v$-$w$~paths~$P_1,\dots,P_{p-1}$.
 For all~$i\in\set{p-1}$ and~$j\in\set{q}$,
 let~$P_i^j$ denote the copy of path~$P_i$ in~$G^j=(V^J,E^j)$.
 In each path~$P_i^j$,
 replace~$v^j$ with~$v^x$ and~$w^j$ with~$w^y$.
 By this,
 we obtain~$q(p-1)$ mutually internally vertex-disjoint~$v^x$-$w^x$-paths.
 It follows that~$C_p(G)$ is $q(p-1)$-connected~\cite{menger1927}.
 \lqed
\end{proof}

We are set to prove the main result of this section.

\begin{proof}[Proof of~\cref{thm:tc:pordham}]
  We prove the statement via induction on~$p\in\N_{\geq 3}$.
  We know that \tcAcr{} is \NP-hard on 3-ordered regular Hamiltonian graphs
  with a Hamiltonian cycle provided.
  Let the statement hold true for all~$p'\in\set{p}$,
  $p\in\N_{\geq 3}$.
  Let~$(G,\HC)$ be an instance of the \NP-hard 
  (by induction) 
  \tcAcr{} on~$p$-ordered regular Hamiltonian graphs
  with Hamiltonian cycle~$\HC$.
  Let~$q\ceq \ceil{22(p+1)/(p-1)}$.
  Compute graph~$G'\ceq C_q(G)$ using~\cref{constr:qcompl} in polynomial time.
  Clearly,
  $G'$ is regular.
  Due to~\cref{obs:qcompl:3col},
  we know that~$G$ is 3-colorable if and only if~$G'$ is 3-colorable.
  Due to~\cref{lem:compl:ham},
  $G'$ is Hamiltonian and we can compute a Hamiltonian cycle~$\HC'$ of~$G'$ in polynomial time.
  Finally,
  due to~\cref{lem:tc:qconn},
  $G'$ is~$q(p-1)\geq 22(p+1)$-connected,
  and hence due to~\cref{fact:conn2ord},
  $G'$ is~$(p+1)$-ordered.
  \lqed
\end{proof}

\section{Conclusion}

\tcTsc{} remains \NP-hard
when requiring a Hamiltonian input graph and a witnessing Hamiltonian cycle in the input,
even if the input graph is already restricted to 4- or 5-regular planar,
to $p$-regular for any~$p\geq 6$,
or to $p$-ordered regular for any~$p\geq 3$ graphs.
We close with the following.

\begin{question}
 Is~\tcTsc{} \NP-hard on~$p$-Hamiltonian-ordered graphs for every~$p\geq 4$?
\end{question}

{
\begingroup
  \renewcommand{\url}[1]{\href{#1}{$\ExternalLink$}}
  \newcommand*{\doi}[1]{\href{http://dx.doi.org/#1}{$\ExternalLink$}}
  \bibliography{../../probs-on-hams-bib}

\newcommand{\noopsort}[1]{}
\begin{thebibliography}{16}
\providecommand{\natexlab}[1]{#1}
\providecommand{\url}[1]{\texttt{#1}}
\expandafter\ifx\csname urlstyle\endcsname\relax
  \providecommand{\doi}[1]{doi: #1}\else
  \providecommand{\doi}{doi: \begingroup \urlstyle{rm}\Url}\fi

\bibitem[Balakrishnan and Ranganathan(2012)]{balakrishnan2012textbook}
Rangaswami Balakrishnan and Kanna Ranganathan.
\newblock \emph{A textbook of graph theory}.
\newblock Springer Science \& Business Media, 2012.

\bibitem[Bollob{\'{a}}s and Thomason(1996)]{BollobasT96}
B{\'{e}}la Bollob{\'{a}}s and Andrew Thomason.
\newblock Highly linked graphs.
\newblock \emph{Comb.}, 16\penalty0 (3):\penalty0 313--320, 1996.
\newblock \doi{10.1007/BF01261316}.
\newblock URL \url{https://doi.org/10.1007/BF01261316}.

\bibitem[Brooks(1941)]{Brooks41}
Rowland~Leonard Brooks.
\newblock On colouring the nodes of a network.
\newblock In \emph{Mathematical Proceedings of the Cambridge Philosophical
  Society}, volume~37, pages 194--197. Cambridge University Press, 1941.
\newblock \doi{10.1017/S030500410002168X}.

\bibitem[Cavallaro and Fluschnik(2021)]{CavallaroF21}
Dario Cavallaro and Till Fluschnik.
\newblock Feedback vertex set on {H}amiltonian graphs.
\newblock \emph{CoRR}, abs/2104.05322, 2021.
\newblock URL \url{https://arxiv.org/abs/2104.05322}.

\bibitem[Dailey(1980)]{Dailey80}
David~P Dailey.
\newblock Uniqueness of colorability and colorability of planar 4-regular
  graphs are {NP}-complete.
\newblock \emph{Discrete Mathematics}, 30\penalty0 (3):\penalty0 289--293,
  1980.
\newblock \doi{https://doi.org/10.1016/0012-365X(80)90236-8}.
\newblock URL
  \url{https://www.sciencedirect.com/science/article/pii/0012365X80902368}.

\bibitem[Diestel(2010)]{Diestel10}
Reinhard Diestel.
\newblock \emph{Graph Theory}, volume 173 of \emph{Graduate Texts in
  Mathematics}.
\newblock Springer, 4th edition, 2010.

\bibitem[Faudree and Faudree(2002)]{FaudreeF02}
Jill~R. Faudree and Ralph~J. Faudree.
\newblock Forbidden subgraphs that imply k-ordered and k-ordered hamiltonian.
\newblock \emph{Discret. Math.}, 243\penalty0 (1-3):\penalty0 91--108, 2002.
\newblock \doi{10.1016/S0012-365X(00)00458-1}.
\newblock URL \url{https://doi.org/10.1016/S0012-365X(00)00458-1}.

\bibitem[Fleischner and Sabidussi(2003)]{FleischnerS03}
Herbert Fleischner and Gert Sabidussi.
\newblock 3-colorability of 4-regular {H}amiltonian graphs.
\newblock \emph{Journal of Graph Theory}, 42\penalty0 (2):\penalty0 125--140,
  2003.
\newblock \doi{10.1002/jgt.10079}.
\newblock URL \url{https://doi.org/10.1002/jgt.10079}.

\bibitem[Fleischner et~al.(2010)Fleischner, Sabidussi, and
  Sarvanov]{FleischnerSS10}
Herbert Fleischner, Gert Sabidussi, and Vladimir~I. Sarvanov.
\newblock Maximum independent sets in 3- and 4-regular {H}amiltonian graphs.
\newblock \emph{Discrete Mathematics}, 310\penalty0 (20):\penalty0 2742--2749,
  2010.
\newblock \doi{10.1016/j.disc.2010.05.028}.
\newblock URL \url{https://doi.org/10.1016/j.disc.2010.05.028}.

\bibitem[Garey et~al.(1976)Garey, Johnson, and Stockmeyer]{GareyJS76}
M.~R. Garey, David~S. Johnson, and Larry~J. Stockmeyer.
\newblock Some simplified {NP}-complete graph problems.
\newblock \emph{Theoretical Computer Science}, 1\penalty0 (3):\penalty0
  237--267, 1976.
\newblock \doi{10.1016/0304-3975(76)90059-1}.
\newblock URL \url{https://doi.org/10.1016/0304-3975(76)90059-1}.

\bibitem[Karp(1972)]{Karp72}
Richard~M. Karp.
\newblock Reducibility among combinatorial problems.
\newblock In \emph{Proceedings of a Symposium on the Complexity of Computer
  Computations}, The {IBM} Research Symposia Series, pages 85--103. Plenum
  Press, New York, 1972.
\newblock \doi{10.1007/978-1-4684-2001-2\_9}.

\bibitem[Kikuno et~al.(1980)Kikuno, Yoshida, and Kakuda]{kikuno1980np}
Tohru Kikuno, Noriyoshi Yoshida, and Yoshiaki Kakuda.
\newblock The {NP}-completeness of the dominating set problem in cubic planer
  graphs.
\newblock \emph{IEICE TRANSACTIONS (1976-1990)}, 63\penalty0 (6):\penalty0
  443--444, 1980.

\bibitem[Menger(1927)]{menger1927}
Karl Menger.
\newblock {\"U}ber regul{\"a}re {{B}}aumkurven.
\newblock \emph{Mathematische Annalen}, 96\penalty0 (1):\penalty0 572--582,
  1927.

\bibitem[Mohar(2001)]{Mohar01}
Bojan Mohar.
\newblock Face covers and the genus problem for apex graphs.
\newblock \emph{Journal of Combinatorial Theory, Series B}, 82\penalty0
  (1):\penalty0 102--117, 2001.
\newblock \doi{10.1006/jctb.2000.2026}.
\newblock URL \url{https://doi.org/10.1006/jctb.2000.2026}.

\bibitem[Mulder(1992)]{Mulder92}
Henry~Martyn Mulder.
\newblock {J}ulius {P}etersen's theory of regular graphs.
\newblock \emph{Discrete Mathematics}, 100\penalty0 (1-3):\penalty0 157--175,
  1992.
\newblock \doi{10.1016/0012-365X(92)90639-W}.
\newblock URL \url{https://doi.org/10.1016/0012-365X(92)90639-W}.

\bibitem[Ng and Schultz(1997)]{NgS97}
Lenhard Ng and Michelle Schultz.
\newblock $k$-ordered {H}amiltonian graphs.
\newblock \emph{Journal of Graph Theory}, 24\penalty0 (1):\penalty0 45--57,
  1997.
\newblock \doi{10.1002/(SICI)1097-0118(199701)24:1<45::AID-JGT6>3.0.CO;2-J}.
\newblock URL
  \url{https://doi.org/10.1002/(SICI)1097-0118(199701)24:1<45::AID-JGT6>3.0.CO;2-J}.

\end{thebibliography}
\endgroup
}

\end{document}